\documentclass[pdflatex,sn-mathphys-num]{sn-jnl}

\usepackage{comment}
\usepackage{graphicx}%
\usepackage{multirow}%
\usepackage{amsmath,amssymb,amsfonts}%
\usepackage{amsthm}%
\usepackage{mathrsfs}%
\usepackage[title]{appendix}%
\usepackage{xcolor}%
\usepackage{textcomp}%
\usepackage{manyfoot}%
\usepackage{booktabs}%
\usepackage{algpseudocode}%
\usepackage{listings}%

\newtheorem{problem}{Problem}
\newtheorem{coro}{Corollary}
\newtheorem{prop}{Proposition}

\newenvironment{algo}[3]{
\flushleft {\bf Algorithm} {\it #1} \\
{\bf Input: } #2 \\
{\bf Output: } #3 \\
}

\newcommand{\ff}[1]{{\mathbb F}_{\! #1}}
\newcommand{\R}{\mathbb R}
\newcommand{\N}{\mathbb N}
\newcommand{\Z}{\mathbb Z}

\newcommand{\cL}{\mathcal L}
\newcommand{\spn}{\mathrm{Span}}
\newcommand{\ev}{{\mathrm ev}}

\newcommand{\vol}{\mathrm{vol}}
\newcommand{\co}[1]{{\mathcal C}_{#1}}

\newcommand{\np}[1]{\mathrm{N}(#1)}
\newcommand{\pyr}[1]{\mathrm{Pyr}(#1)}
\newcommand{\te}[1]{\mathrm{#1}}
\newcommand*\xbar[1]{%
  \hbox{%
    \vbox{%
      \hrule height 0.6pt 
      \kern0.3ex
      \hbox{%
        \kern-0.1em
        \ensuremath{#1}%
        \kern-0.0em
      }%
    }%
  }%
}

\newcommand{\vX}{\xbar{X}}

\theoremstyle{thmstyleone}%
\newtheorem{theorem}{Theorem}
\theoremstyle{thmstyletwo}%
\newtheorem{remark}{Remark}%

\theoremstyle{thmstylethree}%

\begin{document}

\title[List decoding of toric codes]{List decoding of toric codes}


\author[1]{\fnm{Silouanos} \sur{Brazitikos}}\email{silouanb@uoc.gr}

\author*[1]{\fnm{Theodoulos} \sur{Garefalakis}}\email{tgaref@uoc.gr}

\author[1]{\fnm{Eleni} \sur{Tzanaki}}\email{etzanaki@uoc.gr}

\affil[1]{\orgdiv{Department of Mathematics and Applied Mathematics}, \orgname{University of Crete}, \orgaddress{\city{Heraklion}, \postcode{70013}, \country{Greece}}}


\abstract{Polynomial evaluation codes hold a prominent place in coding theory.  In this work,
we study the problem of list decoding for a general class of polynomial evaluation codes,
also known as Toric codes, that are defined for any given convex polytope $P$. Special cases, 
such as Reed-Solomon and Reed-Muller codes, have been studied extensively. We present
a generalization of the Guruswami-Sudan algorithm that takes into account the geometry and the
combinatorics of $P$ and compute bounds for the decoding radius. 
}




\maketitle



\section{Introduction}

  Let $q$ be a power of a prime and $\ff{q}$ the finite field with $q$ elements.
  We consider a lattice polytope $P\subseteq\R^m$ and we denote by $\cL_{q}(P)$
  the space of Laurent polynomials over $\ff{q}$ whose monomials have 
  exponent vectors in $P\cap\Z^m$, 
  that is, $$\cL_{q}(P) = \spn_{\ff{q}}\{X_1^{a_1}\cdots X_m^{a_m}  : (a_1,\ldots,a_m)\in 
  P\cap\Z^m\}.$$
  If $S = \{p_1,\ldots,p_n\}\subseteq (\ff{q}^*)^m$
  then we define the evaluation map 
  %
  \begin{eqnarray*}
    \ev : \cL_q(P) &\longrightarrow& \ff{q}^n \\
    f &\mapsto& \bigl(f(p_1),\ldots,f(p_n)\bigr).
  \end{eqnarray*}
  The {\em toric code} related to the polytope $P$, denoted by $\co{P,q}$, is the 
  image  of the map $\ev$ over all $f\in\cL_{q}(P)$. 
  If the field $\ff{q}$ is clear from the context, we simply write $\co{P}$ suppressing $q$ in the notation. 
  Since the polynomials in $\cL_q(P)$ are evaluated at points in $(\ff{q}^*)^m$
  and $x^a=x^b$ for any $x\in\ff{q}^*$ and $a\equiv b\pmod{q-1}$, we may assume
  that $P\cap\Z^m \subseteq [0,q-2]^m$. The set $S$ of evaluation points is often taken, in the literature,
  to be $(\ff{q}^*)^m$, but this assumption is not essential in what follows. In fact, we assume that 
  $P\cap\Z^m \subseteq [a_1,b_1]\times \cdots \times [a_m,b_m]$
  with $0\leq a_i \leq b_i \leq q-2$ for $1\leq i\leq m$ and the set of points $S$ contains
  a large enough box, that is, $S_1\times \cdots \times S_n\subseteq S$ for sets $S_i\subseteq \ff{q}^*$,
  with $|S_i| > b_i - a_i+1$. The assumptions on $P\cap \Z^m$ and $S$ and the Combinatorial Nullstellensatz
  \cite{alon99} show that the kernel of $\ev$ is trivial, and therefore the dimension
  $k:=\dim(\co{P})$ equals the number of lattice points $|P\cap\Z^m|$ of $P$. The space $\ff{q}^n$ is equipped with the
  Hamming metric $\Delta$. We denote by $d(\co{P})$ the {\em distance} of the code  $\co{P}$, that is, the minimum distance between distinct points of the code. It can be easily checked that
  $$ d(\co{P})= n - \max_{0\neq{f}\in\cL_q(P)} |Z(f)|,$$
  where $Z(f)$ denotes the points in $S$ where $f$ vanishes. 

  Toric codes may be viewed as a generalization of the well-known generalized
  Reed-Solomon (GRS) codes in higher dimensions. Indeed, a GRS code is
  $\co{P}$, where the polytope $P$ is the line segment $[0,k-1]$. Furthermore,
  toric codes are a generalization of Reed-Muller codes, that
  are toric codes related to Simplex polytopes.  
  The generality of their definition, in particular the dimension $m$ of the ambient space of the
  defining polytope and the shape of the polytope itself, does not allow
  for very strong and uniform results, as is the case for GRS codes.
  Thus, progress is made by studying special cases. For instance,
  in \cite{SoprLpCt15} the author computes the distance of codes defined by
  special polytopes and also the distance of codes arising from the
  combinatorial construction of polytopes, such as dilation and
  cross product. In \cite{LitSchen06}, the authors focus on the
  dimension $m=2$ and compute the exact distance for various polygons.
  
  On the algorithmic side, efficient decoding of GRS codes
  has been known for more than fifty years; see for instance
  \cite{gor-zie-61,berle-84,sugi-75}. A major breakthrough in
  the area came in 1997 when M. Sudan discovered a list decoding
  algorithm for GRS codes \cite{Sud_DRS97}. The list decoding 
  problem for the code $\co{P}$ is defined as follows:

\begin{problem}[List Decoding]
  Given the finite field $\ff{q}$, the polytope $P$, the evaluation
  set $S=\{p_1,\ldots, p_n\}$ (that define $\co{P}$), an integer
  $t$ and a point $y=(y_1,\ldots, y_n)\in \ff{q}^n$, compute every
  codeword $c \in \co{P}$ such that $\Delta(c,y)\leq n-t$.
\end{problem}

  It is evident that the list decoding problem may be stated as
  a polynomial reconstruction problem:

\begin{problem}[Polynomial Reconstruction]
  Given the finite field $\ff{q}$, the polytope $P$, the evaluation
  set $S=\{p_1,\ldots, p_n\}$ (that define $\co{P}$), an integer
  $t$ and a point $y=(y_1,\ldots, y_n)\in \ff{q}^n$, compute every
  polynomial $f \in \cL_q(P)$ such that $f(p_i) = y_i$ for at
  least $t$ points $p_i$ of $S$.
  \end{problem}
  
  The algorithm of Sudan was later improved by M. Sudan and V. Guruswami \cite{GuSud_IDRS99}.
  In particular, in \cite{Sud_DRS97}, M. Sudan sketches how his
  method can be generalized to higher dimensions and provides
  bounds for the decoding radius. This idea has been developed 
  further by Pellikaan and Wu \cite{pel-wu04}, and improved by Augot and 
  Stepanov \cite{aug_step09} for list decoding of Reed-Muller codes. 
    
  In this work, we formulate a variant of Sudan's decoding algorithm,
  that generalizes existing algorithms for Generalized Reed-Solomon
  and for Reed-Muller codes. Our algorithm works for any polytope $P$ and takes
  into account the geometry of the underlying polytope. In section~\ref{sec:prelim},
  we present the mathematical background that is needed for the description and the analysis 
  of the algorithms. In section~\ref{sec:basic}, we give the description
  and analysis of the basic method. Although the basic method is
  a special case of the improved method, that is discussed and analyzed in 
  section ~\ref{sec:improved}, we choose to present it first, as it contains
  the main ideas and avoids some of the technicalities of the latter method.
  The main result is Theorem~\ref{thm:improved1}, which uses the so-called Footprint Lemma
  for bounding the number of zeros of a multivariate polynomial.
  
\section{Preliminaries}\label{sec:prelim}

\paragraph{ Minkowski sum and Newton polytopes}

The {\em Minkowski sum} $P+Q$ of two sets 
$P,Q\subseteq \R^m$  is the usual vector sum of 
all pairs of points in $P,Q$,\, that is, $P+Q=\{x+y: x\in P, y\in Q\}.$
It is not hard to see that if $P,Q$ are lattice polytopes then so is their 
Minkowski sum. 

The {\em support} of a polynomial $f(X_1,\ldots,X_m)\in \mathbb K[X_1,\ldots,X_m]$  
(where $\mathbb K$ is any field) is the set of exponent vectors of the 
monomials appearing in $f$. The {\em Newton polytope} $\np{f}\subseteq \mathbb R^m$ of 
the polynomial $f$
is the convex hull of the support of $f$. 
It is a well known result (see, for example, \cite{Ostro76}) that, 
for polynomials $f,g \in \mathbb K[X_1,\ldots,X_m]$
the Newton polytope of their product is the Minkowski sum of their Newton polytopes, 
i.e.,\, $\np{fg}=\np{f}+\np{g}$. This is one of the reasons why the Newton polytope 
can be thought of as a notion of  degree for multivariate polynomials. 

It is possible to compute upper bounds for the number of zeros of a polynomial $f$ in 
a box $S_1\times\cdots\times S_m$ in terms of the sizes $s_j = |S_j|$ and the multi-degree
of $f$ using the so-called footprint bound.  
\begin{theorem}\cite{geil-2000} \label{thm:footprint}
Let $S_j\subseteq \mathbb{K}$ for $1\leq j\leq m$ with $s_j = |S_j|$. 
For a non-zero polynomial $f(X_1,\ldots,X_m) \in \mathbb{K}[X_1,...,X_m]$ let $X_1^{i_1}\cdots X_m^{i_m}$
be a leading monomial and assume $i_1 < s_1,\ldots ,i_m < s_m$. Then $f$
possesses at most $s_1\cdots s_m - (s_1 - i_1)\cdots(s_m - i_m)$ roots over $S_1\times\cdots \times S_m$.
\end{theorem}
Applying the footprint bound to a polynomial $f(X_1,\ldots,X_m)$ with Newton polytope
$\np{f} \subseteq [a_1, b_1]\times\cdots\times [a_m,b_m]$, we obtain the following corollary.
\begin{coro} \label{coro:number-of-zeros}
Let $f(X_1,\ldots,X_m)\in \mathbb{K}[X_1,\ldots,X_m]$ be a non-zero polynomial with 
Newton polytope $\np{f} \subseteq [a_1,b_1]\times\cdots\times [a_m,b_m]$. Let $S_j\subseteq \mathbb{K}^*$
for $1\leq j\leq m$ with $s_j = |S_j| > \ell_j + 1$, where $\ell_j = b_j - a_j$
and $S\subseteq (\mathbb{K}^*)^m$ with $S_1\times \cdots\times S_m \subseteq S$.
The number of zeros of $f$ in $S$ is upper bounded by 
$$ |S| - \prod_{j=1}^m (s_j - \ell_j). $$
\end{coro}
\begin{proof}
The assumption $\np{f} \subseteq [a_1,b_1]\times\cdots\times [a_m,b_m]$ implies that
\[ f(X_1,\ldots, X_m) = X_1^{a_1}\cdots X_m^{a_m} g(X_1,\ldots, X_m),\]
where $g$ is a polynomial with Newton polytope contained in $[0,\ell_1]\times\cdots\times [0,\ell_m]$.
Since $S\subseteq (\mathbb{K}^*)^m$, $f$ and $g$ have the same zeros in $S$.   
The leading monomial of $g$ is $X_1^{i_1}\cdots X_m^{i_m}$ with $0\leq i_j\leq \ell_j$ for $1\leq j\leq m$
and its number of zeros in $S_1\times\cdots\times S_m$ is bounded, using Theorem~\ref{thm:footprint}, by 
\[ s_1\cdots s_m - \prod_{j=1}^m(s_j - i_j) \leq s_1\cdots s_m - \prod_{j=1}^m(s_j - \ell_j). \]
Thus the number of zeros in $S$ is at most 
\[ |S| - s_1\cdots s_m + s_1\cdots s_m - \prod_{j=1}^m(s_j - \ell_j)
   = |S| - \prod_{j=1}^m (s_j - \ell_j).
\]
\end{proof}

For univariate polynomials, it is well known that the degree of a non-zero polynomial 
is an upper bound for its number of roots, counted with multiplicity. The analog of this
for multivariate polynomials also holds, as was shown by Augot and Stepanov.
\begin{theorem}[\cite{aug_step09}, Lemma 1] \label{schwarz-zippel}
Let $f(X_1,\ldots,X_m)\in \mathbb{K}[X_1,\ldots,X_m]$ be a polynomial of total degree $d$ 
and $S\subseteq \mathbb{K}$, with $|S|=s$. The sum of multiplicities of $f(X_1,\ldots, X_m)$
over the points in $S^m$ is at most $d s^{m-1}$.
\end{theorem}

The footprint lemma has also been generalized to take into account multiplicities of roots.

\begin{theorem}[\cite{geil-thomsen-2015}, Theorem 17] \label{thm:footprint-multi}
Let $f(X_1,\ldots,X_m)\in \mathbb{K}[X_1,\ldots,X_m]$ with leading monomial $X_1^{i_1}\cdots X_m^{i_m}$,
and $S_i\subseteq \mathbb{K}^*$, $|S_i|=s_i$ for $1\leq i\leq m$.
Assume that $r\in\N$ satisfies
\begin{enumerate}
    \item $i_m < r s_m$,
    \item $i_j < s_j \cdot \min \left\lbrace \frac{\sqrt[m-1]{r}-1}{\sqrt[m-1]{r}-\frac{1}{r}}, \frac{\sqrt[m-2]{2}-1}{\sqrt[m-2]{2} - \frac{1}{2}}\right\rbrace$, for $1\leq j\leq m-1$.
\end{enumerate}
Then the number of zeros of $f(X_1,\ldots, X_m)$ in $S_1\times\cdots\times S_m$ with multiplicities
at least $r$ is at most
 \[ s_1\cdots s_m - s_1\cdots s_m \left(s_1-\frac{i_1}{r}\right)\cdots \left(s_m - \frac{i_m}{r}\right).\]
\end{theorem}

\paragraph{Ehrhart Polynomials}
Let $P$ be an $m$-dimensional polytope in $\R^m$. We denote by
$\lambda P = \{\lambda p : p \in P\}$,
the $\lambda$-dilation of $P$. The Ehrhart polynomial of $P$ is the 
function $L_P(\lambda)=|\lambda P\cap \Z^m|$ which counts the integer points
of $\lambda P$. It is well known that if $P$ is an integral polytope -- that is, all its vertex coordinates are integers -- $L_P(\lambda)$ is a polynomial in $\lambda$
of degree $m$ whose leading coefficient equals the volume of $P$.
For any polytope $P\subseteq \R^m$ we define the pyramid $\pyr{P}\subseteq\R^{m+1}$ over 
$P$ as the polytope obtained by taking the convex hull of 
$(0,0,\ldots, 0,1)$ and all vertices of $P$ embedded in the  hyperplane $x_{m+1}=0$ of $\R^{m+1}$,\, 
i.e.,
 $\pyr{P}=\te{conv}\{(0,0,\ldots, 0,1),({\bf v},0): {\bf v} \text{ vertex of } P\}$.
In the next proposition, we make a connection between the Ehrhart polynomial of $P$ and the Ehrhart polynomial of $\pyr{P}$.  

\begin{prop}\label{prop:ehrhart-pyr}
Let \(P\subset\mathbb R^m\) be an \(m\)-dimensional lattice polytope. Then
\[
L_{\pyr{P}}(\lambda)=\sum_{k=0}^{\lambda}L_P(k).
\]
Moreover, for every \(\lambda\ge m\),
\[
L_{\pyr{P}}(\lambda)>m!\vol(P)\binom{\lambda+1}{m+1}.
\]
Consequently, if \(n,r\in\mathbb N\), then \(L_{\pyr{P}}(\lambda)>n\) whenever
\[
\lambda\ge 
\left\lceil 
\left(\frac{(m+1)n}{\vol(P)}\right)^{\frac1{m+1}}
+m-1
\right\rceil ,
\]
and \(L_{\pyr{P}}(\lambda)>n\binom{m+r}{m+1}\) whenever
\[
\lambda\ge 
\left\lceil 
\left(
\frac{(m+1)n}{\vol(P)}\binom{m+r}{m+1}
\right)^{\frac1{m+1}}
+m-1
\right\rceil .
\]
\end{prop}

\begin{proof}
For \(a=0,\ldots,\lambda\), the section of \(\lambda\pyr{P}\) by the hyperplane
\(x_{m+1}=a\) is \((\lambda-a)P\). Hence
\[
L_{\pyr{P}}(\lambda)
=
\sum_{a=0}^{\lambda}L_P(\lambda-a)
=
\sum_{k=0}^{\lambda}L_P(k).
\]
Let \(h_0^*,\ldots,h_m^*\) be the \(h^*\)-vector of \(P\). Then
\(h_i^*\ge0\), \(\sum_i h_i^*=m!\vol(P)\), and
\(L_P(k)=\sum_{i=0}^m h_i^*\binom{k+m-i}{m}\). Therefore,
\[
L_{\pyr{P}}(\lambda)
=
\sum_{i=0}^m h_i^*\binom{\lambda+m+1-i}{m+1}.
\]
If \(\lambda\ge m\), then
\(\binom{\lambda+m+1-i}{m+1}\ge \binom{\lambda+1}{m+1}\) for every
\(i=0,\ldots,m\), and the inequality is strict for \(i=0\). Since \(h_0^*=1\),
we get
\[
L_{\pyr{P}}(\lambda)
>
\binom{\lambda+1}{m+1}\sum_{i=0}^m h_i^*
=
m!\vol(P)\binom{\lambda+1}{m+1}.
\]
Finally,
\[
m!\vol(P)\binom{\lambda+1}{m+1}
=
\frac{\vol(P)}{m+1}\prod_{j=0}^{m}(\lambda+1-j).
\]
Thus, if \(\lambda\ge A+m-1\), then each factor \(\lambda+1-j\) is at least \(A\),
and the last expression is at least \(\frac{\vol(P)}{m+1}A^{m+1}\). Taking
\(A=\left(\frac{(m+1)n}{\vol(P)}\right)^{1/(m+1)}\) gives the first consequence,
and taking
\[
A=\left(
\frac{(m+1)n}{\vol(P)}\binom{m+r}{m+1}
\right)^{1/(m+1)}
\]
gives the second.
\end{proof}

\begin{prop} \label{prop:support-F}
Let $I=\lambda \,{\rm Pyr}(P)$, and
\[ 
Q(\vX, Y)=\sum\limits_{(i_1,\ldots,i_m,k)\,\in\, I}q_{i_1,\ldots,i_m,k} X_1^{i_1}\cdots X_m^{i_m} Y^k.
\]
Then the support of every polynomial $F(\vX)=Q(\vX,f(\vX))$ is contained in 
$\lambda\, P$. In other words, $F(\vX)\in \cL_q(\lambda\, P)$. 
\end{prop}
\begin{proof}
We have
\begin{align}
 F(\vX) =Q(\vX,f(\vX))  
       &=\sum\limits_{(i_1,\ldots,i_m,k)\in I}q_{i_1,\ldots,i_m,k} 
       X_1^{i_1}\cdots X_m^{i_m} f(X_1,\ldots,X_m)^k \notag \\
       & =\sum_{k=0}^{\lambda}\sum_{(i_1,\ldots, i_m)\in (\lambda-k)P } q_{i_1,\ldots,i_m,k}
       X_1^{i_1}\cdots X_m^{i_m} f(X_1,\ldots, X_m)^{k}. \label{monom}
\end{align}
To compute the support of $F(\vX)$ we consider the monomials appearing in each 
$\sum\limits_{(i_1,\ldots, i_m)\in (\lambda-k)P }q_{i_1,\ldots,i_m,k} X_1^{i_1}\cdots X_m^{i_m} f(X_1,\ldots,X_m)^{k}$
of \eqref{monom}. Since the support  of each $X_1^{i_1}\cdots X_m^{i_m}$ lies in $(\lambda-k)P$
and that of $f(X_1,\ldots, X_m)^{k}$ lies in $kP$, the support of their
product lies in the Minkowski sum $(\lambda-k)P+kP=\lambda P$. 
This holds for all $0\leq k\leq \lambda$, which further implies that the support 
of \eqref{monom} lies in $\lambda P$, as well. 
\end{proof}

\section{Basic Method} \label{sec:basic}
We are given an evaluation code $\co{P}$, a point $y = (y_1,\ldots,y_n)\in \ff{q}^n$
and $t\in \N$. Our task is to compute all polynomials $f\in \cL_q(P)$, such 
that $\Delta(\ev(f), y) \leq n-t$ (equivalently, $f(p_i)=y_i$ for at least $t$ points $p_i \in S$).

For brevity, we denote by $\vX$ the ``vector" of variables $(X_1,\ldots, X_m)$.
Following the work of Sudan \cite{Sud_DRS97}, our strategy is to construct an 
auxiliary non-zero polynomial $Q(\vX,Y) \in \ff{q}[\vX,Y]$ with the property:
\begin{align}
f(\vX)\in \cL_{q}(P) \mbox{\, and \,} \Delta(\ev(f), y)\leq n-t  \Longrightarrow\ 
Q(\vX,f(\vX)) \equiv 0.
\label{Sudan}
\end{align}
Given such a polynomial $Q(\vX,Y)$, the polynomials $f(\vX)\in \cL_{q}(P)$
with $\Delta(\ev(f), y)\leq n-t$ can be computed as roots of $Q(\vX,Y)$, viewed as
a polynomial in $\ff{q}(\vX)[Y]$. We note, that $Q(\vX, Y)$ may have roots $g(\vX)$
that do not satisfy the required conditions. It is an easy computational task
to check those conditions for each root of $Q(\vX, Y)$.

To construct $Q(\vX,Y)$ we write
\begin{align}
  Q(\vX,Y) = \sum_{(i_1,\ldots,i_m,k)\in I} q_{i_1,\ldots,i_m,k}\, X_1^{i_1}\cdots X_m^{i_m} Y^k \in \ff{q}[\vX,Y]
  \label{polQ}
\end{align}
where $I\subseteq\Z^{m+1}$ is the support of $Q$. The algorithm works in two stages. Note
that the parameter $t$ is given implicitly, as part of the index set $I$.

\begin{algo}{Basic Method}
{Polytope $P$, set of points $S=\{p_1,\ldots,p_n\}\subseteq (\ff{q}^*)^m$, 
point $y=(y_1,\ldots,y_n)\in \ff{q}^n$, index set $I\subseteq \Z_{\geq 0}^{m+1}$.}
{Every $f\in \cL_q(P)$ such that $f(p_j) = y_j$ for at least $t$ points in $S$.}
\begin{enumerate}
\item Compute a non-zero solution of the linear system
\begin{equation} \label{eq:defQ}
  Q(p_{j},y_{j}) = 0,\ \ 1\leq j\leq n. 
\end{equation}
\item Compute the roots of $Q(\vX, Y)$, viewed as a polynomial in $\ff{q}(\vX)[Y]$,
and output the roots that lie in $\cL_q(P)$ and satisfy $f(p_j)=y_j$ for at least
$t$ points of $S$.
\end{enumerate}
\end{algo}

\begin{theorem} \label{thm:basic-method}
Let $P\subseteq \R^m$ be a lattice polytope, $\ff{q}$ be the finite field with $q$ elements,
$S=\{p_1,\ldots, p_n\}\subseteq (\ff{q}^*)^m$ and let $\co{P}$ be the related evaluation code.
Let $Q(\vX, Y)$ be the polynomial defined in Equation \ref{polQ}. Assume 
\begin{enumerate}
\item $I\subseteq\Z_{\geq 0}^{m+1}$, with $|I| > n$,
\item For every $f(\vX)\in \cL_q(P)$, if $Q(\vX,f(\vX))$ is not identically zero,
then it has less than $t$ zeros in $S$.
\end{enumerate}
Then the {\it Basic Method} solves the Polynomial Reconstruction Problem using $O(|I|^3)$ operations in $\ff{q}$.
\end{theorem}
\begin{proof}
The assumption $|I|>n$ ensures that a non-zero solution of \eqref{eq:defQ} exists.
Let $f\in \cL_q(P)$ be a polynomial with $f(p_{j}) = y_{j}$ for at least $t$
points $p_{j}\in S$. For those points we have $Q(p_{j},f(p_{j})) = Q(p_{j},y_{j}) = 0$
or, equivalently, that the polynomial $F(\vX) = Q(\vX,f(\vX))$ has at least $t$ zeros in $S$. 
If $F(\vX)$ were not identically zero, the second assumption would imply that
$F(\vX)$ has less than $t$ zeros in $S$, a contradiction. Hence $F(\vX)\equiv 0$. Equivalently, 
$f(\vX)$ is a root of the polynomial $Q(\vX, Y)\in \ff{q}(\vX)[Y]$.

Regarding the time complexity of the algorithm, step 1 amounts to solving a linear system
of $n$ equations and $|I|$ variables. This can be done
with $O(|I|^3)$ operations in $\ff{q}$ using standard Gauss elimination.
In step 2, the roots of $Q(\vX, Y)$ can be computed using the
algorithm in \cite{Wu2002}, using $O(N^3)$ operations in $\ff{q}$, as shown in \cite{pel-wu04},
where $N$ is the number of terms in $Q(\vX, Y)$. As shown above, $N\leq |I|$ and the total time complexity is as claimed.
\end{proof}

The crucial parameter in this approach is the index set $I$, which has to
be chosen so that 
\begin{enumerate}
    \item $|I| > n$, and 
    \item every polynomial $Q(\vX, f(\vX))$ for $f\in \cL(P)$, that is not 
    identically zero, has less than $t$ roots in $S$.
\end{enumerate}

The vital difference when comparing to Sudan's method is the 
fact that, unlike the case of univariate polynomials, one cannot always compute
tight  upper bounds for the number of roots of multivariate polynomials. 
In fact, the number of roots of a multivariate polynomial is  
strongly  related to the geometry of its support or, equivalently, the 
geometry of its Newton polytope.

Since the evaluation points lie in $(\ff{q}^*)^m$, multiplying all polynomials in
\(\cL_q(P)\) by a fixed Laurent monomial does not change their zero sets on \(S\).
Thus, after translating \(P\) by a lattice vector, we may assume that
\(
P\subseteq \mathbb R_{\geq 0}^m.
\)
We apply the method outlined in Theorem~\ref{thm:basic-method}, for 
\(
I=(\lambda\pyr{P})\cap\Z_{\geq 0}^{m+1},
\)
where \(\lambda\in\N\) is a parameter to be specified later. For such a \(\lambda\), set
\[
\Phi_\lambda(P;s_1,\ldots,s_m)
=
\min_{\alpha\in \lambda P\cap\Z^m}
\prod_{j=1}^m(s_j-\alpha_j).
\]

\begin{theorem}\label{thm:basic1}
Let \(P\subseteq \R^m\) be a lattice polytope, \(\ff{q}\) be the finite field with
\(q\) elements, \(S=\{p_1,\ldots,p_n\}\subseteq(\ff{q}^*)^m\), and let \(\co{P}\) be
the related evaluation code. Assume that \(S\) contains the box
\(S_1\times\cdots\times S_m\), with \(|S_j|=s_j\), and that
\[
\alpha_j<s_j
\qquad\text{for every }\alpha\in \lambda P\cap\Z^m
\text{ and every }1\le j\le m.
\]
Then there exists a polynomial-time algorithm that solves the polynomial
reconstruction problem for every integer \(t\) such that
\[
n<\sum_{k=0}^{\lambda}L_P(k)
\]
and
\[
n-\Phi_\lambda(P;s_1,\ldots,s_m)<t.
\]
\end{theorem}

\begin{proof}
We apply Theorem~\ref{thm:basic-method} for
\(I=(\lambda\pyr{P})\cap\Z_{\geq 0}^{m+1}\). By Proposition~\ref{prop:ehrhart-pyr},
\[
|I|=|\lambda\pyr{P}\cap\Z^{m+1}|=\sum_{k=0}^{\lambda}L_P(k)>n.
\]
Next, we bound the number of zeros of
\[
F(\vX)=Q(\vX,f(\vX))\in\ff{q}[\vX],
\]
where \(f\in\cL_q(P)\). By Proposition~\ref{prop:support-F}, the support of \(F\)
is contained in \(\lambda P\). If \(F\not\equiv0\), let \(\vX^\alpha\) be a leading
monomial of \(F\). Then \(\alpha\in\lambda P\cap\Z^m\). Hence, by the footprint
bound, \(F\) has at most
\[
s_1\cdots s_m-\prod_{j=1}^m(s_j-\alpha_j)
\]
zeros in the box \(S_1\times\cdots\times S_m\). Since all points of
\(S\setminus(S_1\times\cdots\times S_m)\) may also be zeros, the number of zeros
of \(F\) in \(S\) is at most
\[
n-\prod_{j=1}^m(s_j-\alpha_j)
\leq n-\Phi_\lambda(P;s_1,\ldots,s_m).
\]
Therefore, if \(F\) has at least \(t\) zeros in \(S\) and
\[
n-\Phi_\lambda(P;s_1,\ldots,s_m)<t,
\]
then \(F\) must be identically zero. Thus \(f(\vX)\) is a zero of \(Q(\vX,Y)\),
viewed as a polynomial in \(Y\) over \(\ff{q}(\vX)\).
\end{proof}
\begin{remark}\label{rem:footprint-vertices}
Assume that \(P\subseteq [0,\ell_1]\times\cdots\times[0,\ell_m]\) and that
\(\lambda\ell_j<s_j\) for every \(1\leq j\leq m\). Then the footprint profile
can be computed on the vertices of \(P\):
\[
\Phi_\lambda(P;s_1,\ldots,s_m)
=
\min_{v\in\operatorname{Vert}(P)}
\prod_{j=1}^m(s_j-\lambda v_j).
\]
Indeed, the function
\[
x\mapsto \sum_{j=1}^m \log(s_j-\lambda x_j)
\]
is concave on \(P\). Hence its minimum over \(P\) is attained at a vertex.
Since \(P\) is a lattice polytope and \(\lambda\in\mathbb N\), the points
\(\lambda v\), \(v\in\operatorname{Vert}(P)\), belong to
\(\lambda P\cap\mathbb Z^m\), and the claim follows.
\end{remark}
\paragraph{Examples.}
We illustrate the gain obtained by using the footprint profile instead of the
projection-length bound or the Schwartz--Zippel bound. For simplicity, in the
examples below we take \(S=S_1\times\cdots\times S_m\) and
\(|S_1|=\cdots=|S_m|=s\).

First, let \(m=2\) and consider the truncated box
\[
P=[0,d]^2\cap\{x_1+x_2\leq d+c\},
\qquad 0<c<d.
\]
Equivalently,
\[
P=\operatorname{conv}\{(0,0),(d,0),(d,c),(c,d),(0,d)\}.
\]
Assume that \(\lambda d<s\). Since
\[
x\mapsto \log(s-\lambda x_1)+\log(s-\lambda x_2)
\]
is concave, the minimum of the footprint is attained at a vertex. Hence
\[
\Phi_\lambda(P;s,s)
=
(s-\lambda d)(s-\lambda c).
\]
Therefore the footprint bound gives
\[
|Z(F)\cap S|
\leq
s^2-(s-\lambda d)(s-\lambda c)
=
\lambda(d+c)s-\lambda^2dc.
\]
The projection-length bound only uses \(\ell_1=\ell_2=d\), and gives
\[
|Z(F)\cap S|
\leq
s^2-(s-\lambda d)^2.
\]
Thus the improvement over the projection-length bound is
\[
\lambda(d-c)(s-\lambda d)>0.
\]
Moreover, every polynomial with support in \(\lambda P\) has total degree at most
\(\lambda(d+c)\). Hence Schwartz--Zippel gives
\[
|Z(F)\cap S|\leq \lambda(d+c)s,
\]
whereas the footprint estimate improves this by the positive term \(\lambda^2dc\).

As a second example, let
\[
P=P_{m,r,d}
=
\{x\in\mathbb R^m:\ 0\leq x_j\leq d,\ 1\leq j\leq m,\ 
x_1+\cdots+x_m\leq rd\},
\]
where \(2\leq r<m\). This is a scaled down-closed hypersimplex. It is neither a
box nor a simplex. Its vertices are the points \(d\mathbf 1_A\), where
\(A\subseteq\{1,\ldots,m\}\) and \(|A|\leq r\). If \(\lambda d<s\), then
\[
\Phi_\lambda(P_{m,r,d};s,\ldots,s)
=
(s-\lambda d)^r s^{m-r}.
\]
Thus the footprint bound gives
\[
|Z(F)\cap S|
\leq
s^m-(s-\lambda d)^r s^{m-r}.
\]
The projection-length bound gives only
\[
|Z(F)\cap S|
\leq
s^m-(s-\lambda d)^m,
\]
which is weaker, since \(r<m\). Finally, every polynomial with support in
\(\lambda P_{m,r,d}\) has total degree at most \(\lambda rd\). Thus
Schwartz--Zippel gives
\[
|Z(F)\cap S|\leq \lambda rd\,s^{m-1}.
\]
On the other hand,
\[
s^m-(s-\lambda d)^r s^{m-r}
=
s^{m-r}\big(s^r-(s-\lambda d)^r\big)
<
\lambda rd\,s^{m-1}.
\]
Therefore the footprint profile is strictly sharper than Schwartz--Zippel in
this family.

Finally, consider the lattice trapezoid
\[
P=\operatorname{conv}\{(0,0),(d,0),(c,h),(0,h)\},
\qquad 0<c<d.
\]
Assume that \(\lambda d<s\), \(\lambda h<s\), and
\[
(d-c)s\leq h(s-\lambda c).
\]
Then the footprint minimum is attained at the vertex \((c,h)\), and therefore
\[
\Phi_\lambda(P;s,s)
=
(s-\lambda c)(s-\lambda h).
\]
Consequently,
\[
|Z(F)\cap S|
\leq
s^2-(s-\lambda c)(s-\lambda h)
=
\lambda(c+h)s-\lambda^2ch.
\]
The projection-length bound uses only the bounding box \([0,d]\times[0,h]\), and
therefore gives
\[
|Z(F)\cap S|
\leq
s^2-(s-\lambda d)(s-\lambda h).
\]
The gain is
\[
\lambda(d-c)(s-\lambda h)>0.
\]
Also, the total degree of every polynomial with support in \(\lambda P\) is at
most
\[
\lambda\max\{d,c+h\}.
\]
Thus Schwartz--Zippel gives
\[
|Z(F)\cap S|
\leq
\lambda\max\{d,c+h\}s.
\]
The footprint estimate gives
\[
\lambda(c+h)s-\lambda^2ch
<
\lambda\max\{d,c+h\}s,
\]
and hence it is again strictly sharper than Schwartz--Zippel.

\begin{remark}
In the examples above, the interpolation parameter \(\lambda\) can be chosen
using the exact Ehrhart polynomial, instead of the volume estimate of
Proposition~\ref{prop:ehrhart-pyr}. For instance, consider the truncated box
\[
P=[0,d]^2\cap\{x_1+x_2\leq d+c\},
\qquad 0<c<d.
\]
By Pick's theorem,
\[
L_P(k)
=
\left(d^2-\frac{(d-c)^2}{2}\right)k^2
+
\frac{3d+c}{2}k
+
1.
\]
Hence the least admissible interpolation parameter is
\[
\lambda_*
=
\min\left\{\lambda\in\mathbb N:
\left(d^2-\frac{(d-c)^2}{2}\right)
\frac{\lambda(\lambda+1)(2\lambda+1)}6
+
\frac{3d+c}{4}\lambda(\lambda+1)
+
(\lambda+1)
>n
\right\}.
\]
For this value of \(\lambda_*\), provided \(\lambda_*d<s\), the footprint profile is
\[
\Phi_{\lambda_*}(P;s,s)
=
(s-\lambda_*d)(s-\lambda_*c).
\]
Thus, in this example, both the interpolation condition and the zero estimate can
be computed explicitly.
\end{remark}

One may use Corollary~\ref{coro:number-of-zeros} instead of
Theorem~\ref{thm:footprint}, to obtain the following weaker consequence.

\begin{coro}\label{thm:basic2}
Let $P\subseteq \R^m$ be a lattice polytope, $\ff{q}$ be the finite field with
$q$ elements, $S=\{p_1,\ldots,p_n\}\subseteq(\ff{q}^*)^m$, and let $\co{P}$ be the
related evaluation code. Denote by $\ell_j$ the length of the projection of $P$
on the $j$-axis. Assume that $S$ contains the box
$S_1\times\cdots\times S_m$, with $|S_j|=s_j$, $1\leq j\leq m$.
Then there exists a polynomial-time algorithm that solves the polynomial
reconstruction problem, for any integer $\lambda$ such that
\[
0<\lambda<\min_{1\leq j\leq m}\frac{s_j}{\ell_j},
\]
and any integer $t$ such that
\[
n<\sum_{k=0}^{\lambda}L_P(k)
\]
and
\[
n-\prod_{j=1}^m(s_j-\lambda\ell_j)<t.
\]
\end{coro}

It is possible to obtain a value of \(\lambda\) and the corresponding \(t\), under reasonable assumptions
on the geometry of the polytope \(P\). Sharper results may be obtained if the polytope is given
and further assumptions are made on the evaluation set \(S\), for instance, if \(S\) is a box and therefore
\(n=s^m\) for some suitably large \(s\).

\begin{theorem}\label{thm:basic3}
    Let \(P\subseteq \R^m\) be a lattice polytope, \(\ff{q}\) be the finite field with \(q\) elements,
    \(S=\{p_1,\ldots, p_n\}\subseteq (\ff{q}^*)^m\) and let \(\co{P}\) be the related evaluation code. 
    Denote by \(\ell_j\) the length of the projection of \(P\) on the \(j\)-axis.
    Assume that \(S\) contains the box \(S_1\times \cdots \times S_m\), with \(|S_j| = s_j > \ell_j\),
    \(1\leq j\leq m\). Let
    \[
    \lambda =
    \left\lceil
    \left(\frac{(m+1)n}{\vol(P)}\right)^{\frac{1}{m+1}}+m-1
    \right\rceil .
    \]
    Further assume that \(\min_{1\leq j\leq m} s_j/\ell_j > \lambda\) and \(\lambda\geq m\).
    Then there exists a polynomial-time algorithm that solves the polynomial reconstruction problem for every integer \(t\) such that
    \[
    t> n-\prod_{j=1}^m(s_j-\lambda \ell_j).
    \]
\end{theorem}

\begin{proof}
    By the choice of \(\lambda\) and Proposition~\ref{prop:ehrhart-pyr}, we have
    \[
    \sum_{k=0}^{\lambda} L_P(k)>n.
    \]
    The existence of the algorithm follows from the assumptions on \(\lambda\) and Corollary~\ref{thm:basic2}.
\end{proof}

\section{Improved Method} \label{sec:improved}
The basic method, outlined in the previous sections, can be improved, following the
work of Guruswami and Sudan \cite{GuSud_IDRS99} and Augot and Stepanov \cite{aug_step09}. Here we
require the points $(p_j, y_j)$, $1\leq j\leq n$ to be zeros of $Q(\vX, Y)$ of multiplicity
at least $r$, where $r$ is a parameter to be determined later. In particular, let
$p_j = (p_{1j},\ldots, p_{mj})$, and 
\[ Q^{(j)}(\vX, Y) = Q(\vX + p_j, Y+y_j),\]
where $Q(\vX, Y)$ is given by Equation~\ref{polQ}. A short calculation shows that
\[ Q^{(j)}(\vX, Y) = 
\sum_{j_1,\ldots,j_m,\nu} q^{(j)}_{j_1,\ldots,j_m,\nu}\ X_1^{j_1}\cdots X_m^{j_m} Y^{\nu},\]
where
\begin{equation} \label{eq:defQj} q^{(j)}_{j_1,\ldots,j_m,\nu} =  
\sum_{\substack{(i_1,\ldots,i_m,k)\in I\\ i_1\geq j_1,\ldots,i_m\geq j_m, k\geq \nu}}
\binom{i_1}{j_1}\cdots\binom{i_m}{j_m}\binom{k}{\nu} p_{1j}^{i_1-j_1}\cdots p_{mj}^{i_m-j_m} y_j^{k-\nu} q_{i_1,\ldots,i_m,k}.
\end{equation}

The polynomial $Q(\vX, Y)$ has a zero at $(p_j, y_j)$ with multiplicity at least $r$ if and only if
$Q^{(j)}(\vX, Y)$ has a zero at $(\bar{0},0)$ of multiplicity at least $r$, that is, if and only if

\[
q^{(j)}_{j_1,\ldots,j_m,\nu} = 0\ \text{ for every }\ j_1,\ldots,j_m, \nu \in \Z_{\geq 0}
\ \text{ such that } \ j_1+\cdots +j_m + \nu < r.
\]

The improved algorithm is the following.
\begin{algo}{Improved Method}
{Polytope $P$, set of points $S=\{p_1,\ldots,p_n\}\subseteq (\ff{q}^*)^m$, 
point $y=(y_1,\ldots,y_n)\in \ff{q}^n$, integer $t$, index set $I\subseteq \Z_{\geq 0}^{m+1}$, $r \in \N$.}
{Every $f\in \cL_q(P)$ such that $f(p_j) = y_j$ for at least $t$ points in $S$.}
\begin{enumerate}
\item Compute a non-zero solution of the linear system
\begin{eqnarray} 
  q^{(j)}_{j_1,\ldots,j_m,\nu} = 0 &\text{ for }& j_1,\ldots,j_m, \nu \in \Z_{\geq 0}, \label{eq:defQj-system}
  \ j_1+\cdots +j_m + \nu < r,\\
 &\text{ and }& 1\leq j\leq n. \nonumber
\end{eqnarray}
\item Compute the roots of $Q(\vX, Y)$, viewed as a polynomial in $\ff{q}(\vX)[Y]$,
and output the roots that lie in $\cL_q(P)$.
\end{enumerate}
\end{algo}

We note that for $r=1$, the improved method reduces to the basic method of Section~\ref{sec:basic}. 

\begin{theorem} \label{thm:improved-method}
Let $P\subseteq \R^m$ be a lattice polytope, $\ff{q}$ be the finite field with $q$ elements,
$S=\{p_1,\ldots, p_n\}\subseteq (\ff{q}^*)^m$ and let $\co{P}$ be the related evaluation code.
Let $Q(\vX, Y)$ be the polynomial defined in Equations \ref{polQ}, \ref{eq:defQj}, and \ref{eq:defQj-system} for
some $r\in\N$. Assume 
\begin{enumerate}
\item $I\subseteq\Z_{\geq 0}^{m+1}$, with $|I| > \binom{m+r}{m+1}\ n$,
\item For every $f(\vX)\in \cL_q(P)$, if $Q(\vX,f(\vX))$ is not identically zero,
then the sum of the multiplicities of its zeros in $S$ is less than $rt$.
\end{enumerate}
Then the {\it Improved Method} solves the Polynomial Reconstruction Problem with $O(|I|^3)$
operations in $\ff{q}$.
\end{theorem}
\begin{proof}
Equation \ref{eq:defQj-system} defines a linear system in $|I|$ variables and at most $\binom{m+r}{m+1}\ n$ equations.
The assumption $|I|>\binom{m+r}{m+1}\ n$ ensures that a non-zero solution exists. Any such solution defines
a polynomial $Q(\vX,Y)$ that has a zero of multiplicity at least $r$ at $(p_j, y_j)$ for every $1\leq j\leq n$.
Let $f\in \cL_q(P)$ be a polynomial with $f(p_{j}) = y_{j}$ for at least $t$
points $p_{j}\in S$. Each of those points is a zero of multiplicity at least $r$ of the polynomial 
$F(\vX) = Q(\vX,f(\vX))$. Hence the sum of the multiplicities of the zeros of
\(F\) in \(S\) is at least \(rt\). If \(F\) were not identically zero, this would
contradict the second assumption. Therefore \(F\equiv0\). Equivalently, 
$f(\vX)$ is a root of the polynomial $Q(\vX, Y)\in \ff{q}(\vX)[Y]$.

Regarding the time complexity of the algorithm, step 1 amounts to solving a linear system of
$\binom{m+r}{m}\ n$ equations and $|I|$ variables. This can be done
with $O(|I|^3)$ operations in $\ff{q}$ using standard Gauss elimination.
In step 2, the roots of $Q(\vX, Y)$ can be computed using the
algorithm in \cite{Wu2002}, using $O(N^3)$ operations in $\ff{q}$, as shown in \cite{pel-wu04},
where $N$ is the number of terms in $Q(\vX, Y)$. As shown above, $N = |I|$ and the total time complexity is as claimed.
\end{proof}
We apply the method outlined in Theorem~\ref{thm:improved-method}, for
\[
I=(\lambda\pyr{P})\cap\Z_{\geq0}^{m+1},
\]
where \(\lambda,r\in\N\) are parameters to be specified later. For such
\(\lambda\) and \(r\), set
\[
\Phi_{\lambda,r}(P;s_1,\ldots,s_m)
=
\min_{\alpha\in\lambda P\cap\Z^m}
\prod_{j=1}^m\left(s_j-\frac{\alpha_j}{r}\right).
\]

\begin{theorem}\label{thm:improved1}
Let \(P\subseteq\R^m\) be a lattice polytope, \(\ff{q}\) be the finite field with
\(q\) elements, \(S=\{p_1,\ldots,p_n\}\subseteq(\ff{q}^*)^m\), and let \(\co{P}\) be
the related evaluation code. Assume that \(S\) contains the box
\(S_1\times\cdots\times S_m\), with \(|S_j|=s_j\), \(1\leq j\leq m\).

Let \(\lambda,r\in\N\) be such that
\[
\sum_{k=0}^{\lambda}L_P(k)>n\binom{m+r}{m+1}.
\]
Assume moreover that, for every \(\alpha=(\alpha_1,\ldots,\alpha_m)\in
\lambda P\cap\Z^m\), the hypotheses of Theorem~\ref{thm:footprint-multi} are
satisfied, namely
\[
\alpha_m<rs_m
\]
and
\[
\alpha_j
<
s_j\cdot
\min\left\{
\frac{\sqrt[m-1]{r}-1}{\sqrt[m-1]{r}-\frac1r},
\frac{\sqrt[m-2]{2}-1}{\sqrt[m-2]{2}-\frac12}
\right\},
\qquad 1\leq j\leq m-1.
\]
Then the improved method solves the polynomial reconstruction problem for every
positive integer \(t\) such that
\[
t>n-\Phi_{\lambda,r}(P;s_1,\ldots,s_m),
\]
using \(O(n^3 m^{3(r-1)})\) operations in \(\ff{q}\).
\end{theorem}

\begin{proof}
We apply Theorem~\ref{thm:improved-method} for
\[
I=(\lambda\pyr{P})\cap\Z_{\geq0}^{m+1}.
\]
By Proposition~\ref{prop:ehrhart-pyr}, choosing 
\[
\lambda = 
\left\lceil 
\left(
\frac{(m+1)n}{\vol(P)}\binom{m+r}{m+1}
\right)^{\frac1{m+1}}
+m-1
\right\rceil, 
\]
we have
\[
|I|
=
|\lambda\pyr{P}\cap\Z^{m+1}|
=
\sum_{k=0}^{\lambda}L_P(k)
>
n\binom{m+r}{m+1}.
\]
Hence the interpolation system has a non-zero solution.

Let \(f\in\cL_q(P)\), and set
\[
F(\vX)=Q(\vX,f(\vX)).
\]
By Proposition~\ref{prop:support-F}, the support of \(F\) is contained in
\(\lambda P\). Suppose that \(F\not\equiv0\), and let \(\vX^\alpha\) be a leading
monomial of \(F\). Then \(\alpha\in\lambda P\cap\Z^m\). By the assumptions on
\(\alpha\), Theorem~\ref{thm:footprint-multi} gives that the number of zeros of
\(F\) in \(S_1\times\cdots\times S_m\) with multiplicity at least \(r\) is at most
\[
s_1\cdots s_m-
\prod_{j=1}^m\left(s_j-\frac{\alpha_j}{r}\right).
\]
Since all points of \(S\setminus(S_1\times\cdots\times S_m)\) may also be such
zeros, the number of points of \(S\) at which \(F\) has multiplicity at least \(r\)
is at most
\[
n-
\prod_{j=1}^m\left(s_j-\frac{\alpha_j}{r}\right)
\leq
n-\Phi_{\lambda,r}(P;s_1,\ldots,s_m).
\]
Therefore, if \(f\) agrees with \(y\) in at least \(t\) points of \(S\), then \(F\)
has multiplicity at least \(r\) for at least \(t\) points. The condition
\[
t>n-\Phi_{\lambda,r}(P;s_1,\ldots,s_m)
\]
forces \(F\equiv0\). Hence \(f(\vX)\) is a root of \(Q(\vX,Y)\), viewed as a
polynomial in \(Y\) over \(\ff{q}(\vX)\).

The complexity bound follows from Theorem~\ref{thm:improved-method}.
\end{proof}

The previous result uses the multiplicity footprint bound and therefore requires
additional restrictions on the exponent vectors in $\lambda P$. A simpler,
although in general weaker, bound can be obtained directly from
Theorem~\ref{schwarz-zippel}. Its advantage is that it applies to every lattice
polytope without any further assumptions on the geometry of $\lambda P$.

\begin{theorem}\label{thm:improved-SZ}
Let $P\subseteq\R^m$ be a lattice polytope, $\ff{q}$ be the finite field with
$q$ elements, $S=\{p_1,\ldots,p_n\}\subseteq(\ff{q}^*)^m$, and let $\co{P}$ be
the related evaluation code. As above, after translating $P$ by a lattice
vector, assume that
\[
P\subseteq\R_{\geq0}^m.
\]
Let $\lambda,r\in\N$ be such that
\[
\sum_{k=0}^{\lambda}L_P(k)
>
n\binom{m+r}{m+1}.
\]
Then the improved method solves the polynomial reconstruction problem for every
positive integer $t$ such that
\[
rt>
\lambda (q-1)^{m-1}
\max_{(i_1,\ldots,i_m)\in P}
(i_1+\cdots+i_m).
\]
Equivalently, it suffices that
\[
t>
\frac{\lambda}{r}(q-1)^{m-1}
\max_{(i_1,\ldots,i_m)\in P}
(i_1+\cdots+i_m).
\]
\end{theorem}

\begin{proof}
We apply Theorem~\ref{thm:improved-method} with
\[
I=(\lambda\pyr{P})\cap\Z_{\geq0}^{m+1}.
\]
By Proposition~\ref{prop:ehrhart-pyr},
\[
|I|
=
|\lambda\pyr{P}\cap\Z^{m+1}|
=
\sum_{k=0}^{\lambda}L_P(k)
>
n\binom{m+r}{m+1},
\]
and therefore the interpolation system has a non-zero solution.

Let $f\in\cL_q(P)$ and set
\[
F(\vX)=Q(\vX,f(\vX)).
\]
By Proposition~\ref{prop:support-F}, the support of $F$ is contained in
$\lambda P$. Hence, if $F\not\equiv0$, then
\[
\deg F
\leq
\lambda
\max_{(i_1,\ldots,i_m)\in P}
(i_1+\cdots+i_m).
\]
Applying Theorem~\ref{schwarz-zippel} to $F$ on
$(\ff{q}^*)^m$, we obtain
\[
\sum_{p\in(\ff{q}^*)^m}\operatorname{mult}_p(F)
\leq
\lambda (q-1)^{m-1}
\max_{(i_1,\ldots,i_m)\in P}
(i_1+\cdots+i_m).
\]
In particular,
\[
\sum_{p\in S}\operatorname{mult}_p(F)
\leq
\lambda (q-1)^{m-1}
\max_{(i_1,\ldots,i_m)\in P}
(i_1+\cdots+i_m).
\]

On the other hand, if $f(p_j)=y_j$, then $F$ has multiplicity at least $r$
at $p_j$. Therefore, if $f$ agrees with $y$ in at least $t$ points of $S$,
then
\[
\sum_{p\in S}\operatorname{mult}_p(F)\geq rt.
\]
Thus, if
\[
rt>
\lambda (q-1)^{m-1}
\max_{(i_1,\ldots,i_m)\in P}
(i_1+\cdots+i_m),
\]
the assumption $F\not\equiv0$ leads to a contradiction. Hence
\[
Q(\vX,f(\vX))\equiv0,
\]
and $f(\vX)$ is a root of $Q(\vX,Y)$, viewed as a polynomial in $Y$ over
$\ff{q}(\vX)$.
\end{proof}

Theorem~\ref{thm:improved-SZ} gives the sharpest bound obtainable from this
argument when the exact Ehrhart polynomial of $P$ is known. A bound that depends
only on the volume of $P$ follows immediately from
Proposition~\ref{prop:ehrhart-pyr}.

\begin{coro}\label{coro:improved-SZ-volume}
Under the assumptions of Theorem~\ref{thm:improved-SZ}, let $r\in\N$ and set
\[
\lambda
=
\left\lceil
\left(
\frac{(m+1)n}{\vol(P)}
\binom{m+r}{m+1}
\right)^{\frac{1}{m+1}}
+m-1
\right\rceil.
\]
Then the improved method solves the polynomial reconstruction problem for every
positive integer $t$ such that
\[
t>
\frac{(q-1)^{m-1}}{r}
\max_{(i_1,\ldots,i_m)\in P}
(i_1+\cdots+i_m)
\left\lceil
\left(
\frac{(m+1)n}{\vol(P)}
\binom{m+r}{m+1}
\right)^{\frac{1}{m+1}}
+m-1
\right\rceil.
\]
\end{coro}

\begin{proof}
By Proposition~\ref{prop:ehrhart-pyr}, the above choice of $\lambda$ satisfies
\[
\sum_{k=0}^{\lambda}L_P(k)
>
n\binom{m+r}{m+1}.
\]
The result therefore follows directly from
Theorem~\ref{thm:improved-SZ}.
\end{proof}

The dependence on the geometry of $P$ becomes particularly transparent when
$r$ tends to infinity.

\begin{coro}\label{coro:improved-SZ-asymptotic}
Under the assumptions of Theorem~\ref{thm:improved-SZ}, for every
\[
t>
(q-1)^{m-1}
\left(
\frac{n}{m!\vol(P)}
\right)^{\frac{1}{m+1}}
\max_{(i_1,\ldots,i_m)\in P}
(i_1+\cdots+i_m),
\]
there exists $r\in\N$ sufficiently large such that the improved method solves
the polynomial reconstruction problem for this value of $t$.
\end{coro}

\begin{proof}
For every $r\in\N$, choose
\[
\lambda_r
=
\left\lceil
\left(
\frac{(m+1)n}{\vol(P)}
\binom{m+r}{m+1}
\right)^{\frac{1}{m+1}}
+m-1
\right\rceil.
\]
Since
\[
\binom{m+r}{m+1}
\sim
\frac{r^{m+1}}{(m+1)!}
\qquad\text{as }r\to\infty,
\]
we have
\[
\frac{\lambda_r}{r}
\longrightarrow
\left(
\frac{n}{m!\vol(P)}
\right)^{\frac{1}{m+1}}.
\]
The result now follows from Corollary~\ref{coro:improved-SZ-volume}.
\end{proof}

In particular, when the evaluation set is the full torus, the preceding result
takes a simple form.

\begin{coro}\label{coro:improved-SZ-full-torus}
Assume that
\[
S=(\ff{q}^*)^m,
\qquad
n=(q-1)^m.
\]
Then, asymptotically as the multiplicity parameter $r$ tends to infinity, the
improved method list decodes from every number of errors smaller than
\[
n\left[
1-
\frac{
\displaystyle
\max_{(i_1,\ldots,i_m)\in P}(i_1+\cdots+i_m)
}{
\displaystyle
\bigl(m!\vol(P)(q-1)^m\bigr)^{\frac{1}{m+1}}
}
\right].
\]
Equivalently, the relative decoding radius approaches
\[
1-
\frac{
\displaystyle
\max_{(i_1,\ldots,i_m)\in P}(i_1+\cdots+i_m)
}{
\displaystyle
\bigl(m!\vol(P)(q-1)^m\bigr)^{\frac{1}{m+1}}
}.
\]
\end{coro}

\begin{remark}
For the simplex
\[
P=d\Delta_m
=
\left\{
(x_1,\ldots,x_m)\in\R_{\geq0}^m:
x_1+\cdots+x_m\leq d
\right\},
\]
we have
\[
\max_{(i_1,\ldots,i_m)\in P}(i_1+\cdots+i_m)=d
\]
and
\[
m!\vol(P)=d^m.
\]
Hence Corollary~\ref{coro:improved-SZ-full-torus} gives the relative decoding
radius
\[
1-
\left(\frac{d}{q-1}\right)^{\frac{1}{m+1}}.
\]
For $m=1$, this becomes
\[
1-\sqrt{\frac{d}{q-1}},
\]
which is the usual Guruswami--Sudan decoding radius for Reed--Solomon codes.
\end{remark}

\section{Conclusion}
In this work, we study the list decoding problem for toric codes, which is a broad class
of linear codes, that include Reed-Solomon and Reed-Muller codes. The main objective of the present work is to
introduce a method for constructing, for any given polytope, a list decoding algorithm with good radius.
We present a method that generalizes the algorithms of Sudan and Guruswami \cite{GuSud_IDRS99} and Augot and Stepanov \cite{aug_step09} for arbitrary polytopes and analyse it.

\vspace*{1cm}

{\flushleft \large \bf Statements and Declarations }\\
The authors have no competing insterests.



\bibliographystyle{abbrv}

\end{document}